\documentclass{article}

\usepackage{amsthm}
\usepackage{amsmath}
\usepackage{amsfonts}
\usepackage{xcolor}
\usepackage[ruled]{algorithm2e}
\usepackage{mathrsfs}
\usepackage{graphicx}
\usepackage{float}
\usepackage{geometry}
\usepackage{tcolorbox}
\usepackage{bm}
\definecolor{linkcolor}{RGB}{170,0,0} 
\definecolor{citecolor}{RGB}{44,160,46} 
\definecolor{urlcolor}{RGB}{0,51,153} 
\usepackage[linktocpage,pdfcreator=easychair.cls-3.4,colorlinks=True,urlcolor=urlcolor,citecolor=citecolor,linkcolor=linkcolor]{hyperref}

\theoremstyle{definition}
\newtheorem{definition}{Definition}

\newtheorem{remark}{Remark}

\theoremstyle{remark}

\theoremstyle{plain}
\newtheorem{theorem}{Theorem}

\newtheorem{lemma}{Lemma}
\newtheorem{cor}{Corollary}

\newcommand{\vect}[1]{\overrightarrow{\bm{#1}}}

\newcommand{\poa}{\mathrm{PoA}}
\newcommand{\cost}{\mathrm{cost}}

\title{Inoculation strategies for bounded degree graphs}

\author{
Mason DiCicco\footnote{Computer Science Department, WPI. [mtdicicco@wpi.edu, dreichman@wpi.edu]}
\and
Henry Poskanzer\footnote{hposkanzer@wpi.edu}
\and
Daniel Reichman$^*$
}

\begin{document}

\maketitle

\begin{abstract}
   We analyze a game-theoretic abstraction of epidemic containment played on an undirected graph $G$: each player is associated with a node in $G$ and can either acquire protection from a contagious process or risk infection. After decisions are made, an infection starts at a random node $v$ and propagates through all unprotected nodes reachable from $v$. It is known that the price of anarchy (PoA) in $n$-node graphs can be as large as $\Theta(n)$. Our main result is a tight bound of order $\sqrt{n\Delta}$ on the PoA, where $\Delta$ is the \emph{maximum degree} of the graph. We also study additional factors that can reduce the PoA, such as higher thresholds for contagion and varying the costs of becoming infected vs. acquiring protection.
\end{abstract}

\section{Introduction}

Networks can be conducive to the spread of undesirable phenomena such as infectious diseases, computer viruses, and false information. A great deal of research has been aimed at studying computational challenges that arise when trying to contain a contagious process~\cite{marathe2013computational,babay2022controlling,braunstein2016network}.

One factor that can contribute to the spread of contagion is the discrepancy between \emph{locally} optimal behavior of rational agents and \emph{globally} optimal behavior that minimizes the total cost to the agents in the network. For example, individuals in a computer network may prefer not to install anti-virus software because it is too expensive, whereas a network administrator may prefer to install copies at key points, limiting the distance a virus could spread and the global damage to the network. The former strategy would be considered a \emph{locally optimal} solution if each individual minimizes their own individual cost, whereas the latter strategy would be a \emph{socially optimal} solution if it minimizes the total cost to all individuals in the network.

How much worse locally optimal solutions can be compared to the social optimum can be quantified by the classical game theoretic notions of \emph{Nash equilibria} and \emph{price of anarchy}. (PoA)~\cite{roughgarden2002bad,roughgarden2005selfish}. Informally, given a multiplayer game, a strategy is a Nash equilibrium if no player can improve her utility by unilaterally switching to another strategy. Then, the PoA is the ratio of the total cost of the \emph{worst} Nash equilibrium to the social optimum. The larger the PoA, the larger the potential cost players in the game may experience due to selfish, uncoordinated behavior. Hence, it is of interest to investigate methods of reducing the PoA in games.

We study the PoA of a game-theoretic abstraction of epidemic containment introduced by~\cite{aspnes2006inoculation}. The \emph{inoculation game} is an $n$-player game in which each player is associated with a node in an undirected graph. A player can buy security against infection at a cost of $C > 0$, or they can choose to accept the risk of infection. If a node is infected, its player must pay a cost $L > 0$. After each player has made their decision, an adversary chooses a random starting point for the infection. The infection then propagates through the graph; any unsecured node that is adjacent to an infected node is also infected.

It is known that the PoA of the inoculation game can be as large as $\Omega(n)$, with the $n$-star (a node connected to $n-1$ other nodes) being one example of a network leading to such a PoA\footnote{This is asymptotically the largest possible PoA; it is shown in~\cite{aspnes2006inoculation} that in any $n$-node network the PoA is at most $O(n)$.}. This raises the question of the relationship between graph-theoretic parameters and the PoA of the inoculation game. This question was explicitly mentioned in~\cite{chen2010better} as an interesting direction for future research. Additional properties of the game may also influence the PoA, such as the relative costs of infection and security, or the threshold of infection.

We study the relationship between the aforementioned factors and the PoA. One motivation for our study is that understanding the links between properties of the inoculation game and the PoA may shed light on methods for designing networks that are less susceptible to contagion. It may also shed light on the effectiveness of interventions (e.g., changing the cost of acquiring inoculations) aimed at controlling contagion. Specifically, we examine following questions:

\begin{itemize}
\item Motivated by the relationship between ``superspreaders" and contagion, we analyze how the \emph{maximum degree} influences the PoA, obtaining asymptotically tight upper and lower bounds in terms of the number of nodes $n$ and the maximum degree $\Delta$. We also study the PoA in graphs with certain structural properties, such as planar graphs.

\item One can try to reduce the PoA by changing the values of $C$ and $L$ (e.g., by making inoculations more accessible). Previous results regarding the PoA~\cite{aspnes2006inoculation,chen2010better} typically examined fixed values of $C$ and $L$. In contrast, we analyze the PoA for specific networks for all $C$ and $L$. 

\item \emph{Complex contagion} refers to contagion models where a node becomes infected only if multiple neighbors are infected. We study the case where the threshold for contagion is 2 (as opposed to 1) and provide very simple analysis of the PoA for certain networks in this contagion model.
\end{itemize}

We also record the asymptotic PoA for graphs families such as trees, planar graphs and random graphs. Details can be found in the Appendix \ref{section:other}.

\subsection{Preliminaries}

Following~\cite{aspnes2006inoculation}, we describe the infection model and the multiplayer game we study as well as a useful characterization of Nash equilibria. 

\subsubsection{Inoculation game}

Unless stated otherwise, we consider graphs with $n$ nodes and identify the set of nodes by integers $[n]:=\{1,\cdots, n\}$. An \emph{automorphism} of a graph $G$ is a permutation $\sigma$ of its vertex set $[n]$ such that $i$ is adjacent to $j$ if and only if $\sigma(i)$ is adjacent to $\sigma(j)$. A graph $G=(V,E)$ is \emph{vertex transitive} if for every two vertices $i,j \in V$ there is automorphism $f$ of $G$ such that $f(i)=j$.
 
\begin{definition}[Inoculation game]
The inoculation game is a one-round, $n$-player game, played on an undirected graph $G$. We assume $G$ is a connected graph and each node is a player in the game. Every node $i$ has two possible actions: Inoculate against an infection, or do nothing and risk being infected. Throughout the paper we assign $1$ to the action of inoculating and $0$ the action of not inoculating. We say the cost of inoculation is $C>0$ and the cost of infection is $L>0$. 
\end{definition}

\begin{remark}
We always assume that $C$ and $L$ are constants independent of $n$, unless otherwise stated. In particular, our lower bounds generally require that the \emph{ratio} $C/L=\Theta(1)$. We discuss the relationship between the costs and the PoA in Section \ref{subsection:costs}.
\end{remark}

The \emph{strategy} of each node $i$ is the probability of inoculating, denoted by $a_i \in [0,1]$, and the \emph{strategy profile} for $G$ is represented by the vector $\vect{a} \in [0,1]^n$. If $a_i \in \{0,1\}$, we call the strategy \emph{pure}, and otherwise \emph{mixed}.

\medskip
Note that a mixed strategy is a probability distribution $\mathcal{D}$ over pure strategies. The cost of a mixed strategy profile to individual $i$ is equal to the expected cost over $\mathcal{D}$,
\begin{align*}
\cost_i(\vect{a}) &= C \cdot \Pr(i \text{ inoculates}) + L \cdot \Pr(i \text{ is infected}) \\
&= C \cdot a_i + L \cdot (1-a_i) p_i(\vect{a}),
\end{align*}
where $p_i(\vect{a})$ denotes the probability that $i$ becomes infected given strategy profile $\vect{a}$ \emph{conditioned on $i$ not inoculating}. The total social cost of $\vect{a}$ is equal to the sum of the individual costs,
\[
\cost(\vect{a}) = \sum_{i=1}^{n} \cost_i(\vect{a}) .
\]

\begin{definition}[Attack graph]
Given a strategy profile $\vect{a}$, let $I_{\vect{a}}$ denote the set of secure nodes (nodes which have inoculated). The \emph{attack graph}, which we denote by $G_{\vect{a}}$, is the sub-graph induced by the set of insecure nodes:
\[
G_{\vect{a}}=G \setminus I_{\vect{a}}.
\]
\end{definition}

\begin{figure}[ht]
    \centering
    \includegraphics[width=0.4 \textwidth]{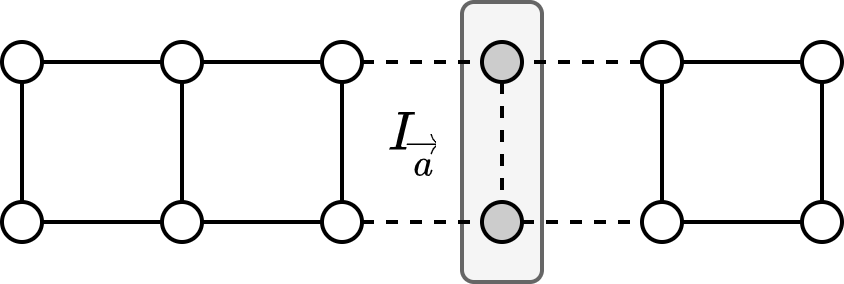}
    \caption{The two vertices in the shaded region are inoculated. The remaining vertices form  the attack graph consisting of two connected components of sizes $4$ and $6$ respectively.}
    \label{fig:grid}
\end{figure}

After every node has decided whether or not to inoculate, a node $i \in V$ is chosen uniformly at random (over all nodes) as the starting point of the infection. If $i$ is not inoculated, then $i$ and every insecure node reachable from $i$ in $G_{\vect{a}}$ are infected. Note that $I_{\vect{a}}$ and $G_{\vect{a}}$ are random variables unless $\vect{a}$ is pure. When the strategies are pure, \cite{aspnes2006inoculation} give the following characterization for the social cost:

\begin{theorem}[\cite{aspnes2006inoculation}]
\label{th:costPure}
Let $\vect{a}$ be a pure strategy profile for a graph $G$. Then,
\[
\cost(\vect{a}) =  C |I_{\vect{a}}| + \frac{L}{n} \sum_{i =1}^\ell k_i^2,
\]
where $k_1,\cdots,k_\ell$ denote the sizes of the connected components in $G_{\vect{a}}$.
\end{theorem}

\subsubsection{Nash equilibria}

\begin{definition}[Nash equilibrium]
\label{def:nash}
A strategy profile $\vect{a}$ is a \emph{Nash equilibrium} (NE) if no nodes can decrease their individual cost by changing their own strategy. 

\medskip
Formally, let $\vect{a^*} := (a_i^*,\vect{a^*_{-i}})$ be a strategy profile where $\vect{a^*_{-i}}$ denotes the strategy profile of all players except node $i$. Then, $\vect{a^*}$ is a Nash equilibrium if, for all $i$,
\[
\cost_i((a_i^*,\vect{a}_{-i}^*)) \leq \cost_i((a_i,\vect{a_{-i}^*})) \text{ for all } a_i \neq a_i^*.
\]
\end{definition}
Similarly to arbitrary strategies, the cost of an NE is simply the sum of expected costs of the individual vertices. 

In the inoculation game, Nash equilibria are characterized by the expected sizes of the \emph{connected components} in the attack graph. 

\begin{theorem}[\cite{aspnes2006inoculation}] 
\label{th:nashchar}
Let $S(i)$ denote the expected size of the component containing node $i$ in the attack graph \textbf{conditioned on $i$ not inoculating}, and let $t = Cn/L$. A strategy $\vect{a}$ is a Nash equilibrium if and only if every node satisfies the following:

\begin{enumerate}
    \item if $a_i=0$, then $S(i) \leq t$
    \item if $a_i=1$, then $S(i) \geq t$
    \item if $0<a_i<1$, then $S(i) = t$ 
\end{enumerate}
\end{theorem}

This follows from the definition of Nash equilibria, recognizing that the threshold on the expected component size, $t=Cn/L$, is the point where the cost of inoculating equals the (expected) cost of not inoculating, $C=LS(i)/n$.

The following upper bound regarding the cost of every Nash equilbrium was observed in~\cite{aspnes2006inoculation}.
\begin{cor}[\cite{aspnes2006inoculation}]
\label{cor:simpleNashUB}
For any graph $G$, every Nash equilibrium has cost at most $\min\{C,L\}n$.
\end{cor}

\begin{proof}
If $C > L$, then the only Nash equilibrium is the strategy $\vect{a} = 0^n$ where no node inoculates, which has cost $Ln$. Otherwise, if $C \leq L$, then the individual cost to any node is at most $C$ (any node will switch its strategy to inoculate, if preferable.)

\end{proof}

\begin{remark}
Later, we show that there exist graphs whose Nash equilibria meet the upper bound of Corollary \ref{cor:simpleNashUB}; for all $C,L>0$, there exists a Nash equilibrium with cost $\min\{C,L\}n$. This property yields bounds on the PoA that are ``stable" with respect to the costs $C,L$.
\end{remark}

\subsubsection{Price of Anarchy}

\begin{definition}[Price of Anarchy]
The Price of Anarchy (PoA) of an inoculation game played on a graph $G$ is equal to the ratio between the cost of the \emph{worst} Nash equilibrium to the cost of the socially optimal strategy\footnote{Observe that as $C,L>0$ the cost is always strictly positive.},
\[
\poa(G) = \frac{\max_{\vect{a} : \text{Nash eq.}} \cost(\vect{a})}{\min_{\vect{a}} \cost(\vect{a})}.
\]
\end{definition}

To upper bound the price of anarchy, we must lower bound the cost of the socially optimal strategy and upper bound the cost of the worst Nash equilibrium. By Corollary \ref{cor:simpleNashUB}, we have the simple upper bound,

\begin{equation}
    \label{eq:poaUB}
    \poa(G) \leq \frac{\min\{C,L\} n}{\min_{\vect{a}} \cost(\vect{a})}
\end{equation}

with equality when there exists a Nash equilibrium with maximum possible cost, $\min\{C,L\}n$.

\subsection{Summary of results}

This section contains the statements for all of our main results. For each, we give a brief description of the conceptual idea behind each proof. 

\subsubsection{Bounding the PoA in terms of the maximum degree}

It was proved in~\cite{aspnes2006inoculation} that the price of anarchy can be as large $\Omega(n)$. Their lower bound is based on the star graph $K_{1,n-1}$ where the optimal strategy (inoculating the root) has cost $O(1)$. Note that inoculating the central node of the star is maximally ``efficient," in that it splits the attack graph into $n-1$ components. This notion of efficiency (i.e., number of components created per inoculation) is the crux of the following result.

\begin{theorem}
\label{th:poaDeltaUB}
Let $G$ be a graph with maximum degree $\Delta$. Then, $\poa(G) = O(\sqrt{n \Delta})$ for all $C,L>0$.
\end{theorem}

\begin{proof}[Proof Idea.]
We make two observations for the sake of lower bounding the social optimum:
\begin{enumerate}
    \item The number of components in the attack graph is bounded above by the number of edges leaving secure nodes. Therefore, if every secure node has degree at most $\Delta$, then there are at most $|I_{\vect{a}}| \Delta$ insecure components. 
    
    \item In the ideal case (i.e., the optimal strategy), all components will have the same size. 
\end{enumerate}
Then, the result follows from a straightforward manipulation of Theorem \ref{th:costPure}.

\end{proof}

We note that bounds on the PoA in terms of the number of nodes and the maximum degree have been stated before without proof. Please see the related work section for more details. 

\medskip
We also show that Theorem \ref{th:poaDeltaUB} is the strongest possible upper bound; for arbitrary values of $n$ and $\Delta$, we can construct a graph $G$ with price of anarchy $\Omega(\sqrt{n \Delta})$. 

\begin{theorem}
\label{th:largerDeltaLB}
For all $n \geq 4\Delta - 2$ and $C \geq L$, there exists a graph $G$ with $\displaystyle \poa(G) = \Omega(\sqrt{n \Delta})$.
\end{theorem}

\begin{proof}[Proof Idea.]
We construct a graph which is ``ideal" with respect to notion of efficiency used in Theorem \ref{th:poaDeltaUB}. In particular, it should be possible to inoculate $\gamma$ nodes such that the attack graph contains $\gamma \Delta$ equally-sized components.

Such a graph is not difficult to come by. For instance, when $\Delta=2$, the cycle graph $C_n$ has this property for $\gamma=\sqrt{n}$; inoculating every $\sqrt{n}$'th node will split the cycle into $\sqrt{n}$ paths of equal length (up to rounding).

\end{proof}

\subsubsection{The relationship between PoA, $C$ and $L$}
\label{subsection:costs}

We have seen that if $C=L$, then the strategy which inoculates no nodes is a Nash equilibrium with cost $Ln$. Can we significantly decrease the cost of the worst-case equilibrium (and PoA) by decreasing $C$, say to $L/2$? We prove that there are graphs for which the worst-case Nash equilibrium has cost $C n$ for any $0<C\leq L$  (i.e., Corollary \ref{cor:simpleNashUB} cannot be improved due to a matching lower bound). This implies that the asymptotic PoA for these graphs remains the same so long as $C/L=\Theta(1)$. We sketch the argument. 

\begin{lemma}
For a graph $G$, let $f(n) := \poa(G)$ for costs $C=L$. Suppose that, if $L$ remains fixed but the cost of inoculation decreases to $C'<L$, then there is a Nash equilibrium in the new game of cost $C' n$. Let $g(n):=\poa(G)$ for costs $C'$ and $L$. If $C'/L=\Theta(1)$, then $g(n)=\Omega(f(n)).$ 
\end{lemma}

\begin{proof}
Reducing the cost of inoculation while keeping $L$ fixed can only decrease the social optimum. On the other hand, by assumption, the cost of the worst case Nash equilibrium for the new game is equal to $C'n$. The result follows as we have seen that, when $C=L$, the cost of the worst case Nash equilibrium is exactly $L n$. 

\end{proof}

It follows that, to establish the asymptotic stability of PoA, it suffices to prove that there is a Nash equilibrium with cost $Cn$ for all $C<L$. By Theorem \ref{th:nashchar}, when $C \geq L$, the pure strategy in which no node inoculates is a Nash equilibrium with cost $Ln$. Similarly, if $C \leq L/n$, then the pure strategy in which every node inoculates is a Nash equilibrium with cost $Cn$. However, if $C \in (L/n,L)$, then it is not clear whether there is a Nash equilibrium with cost $Cn$. We now show that, for certain graphs, \emph{there is} a Nash equilibrium of cost $C n$ for all $0< C < L$ therefore establishing the asymptotic stability of the PoA for such graphs.   

\medskip
We say that a Nash equilibrium is \emph{fractional} if no node has a pure strategy; every node $i$ chooses her action with probability differing from $0$,$1$. 

\begin{lemma}
\label{lem:mixedLB}
The cost of every fractional Nash equilibrium is equal to $C n$.
\end{lemma}

\begin{proof}
Suppose strategy $\vect{a}$ is a Nash equilibrium with $a_i \in (0,1)$ for all $i$. As a consequence of Theorem \ref{th:nashchar}, the expected component size $S(i) = Cn/L$ for all $i$. Thus, probability of infection for any node $i$ is equal to $p_i(\vect{a}) = C/L$. By definition,
\[
\cost(\vect{a}) = \sum_{i=1}^n C \cdot a_i + L \cdot (1-a_i) \frac{C}{L} = C n.
\]
\end{proof}

It is non-trivial to show a fractional equilibrium exists (note that Nash's theorem does not guarantee existence because the space of fractional strategies is not compact). However, it is possible to show that some graphs will always exhibit such an equilibrium.

\begin{theorem}
\label{th:starNash}
Let $G=K_{1,n-1}$ (i.e., the $n$-star). For all $C \in (L/n,L)$, there exists a fractional Nash equilibrium.
\end{theorem}

\begin{proof}[Proof Idea.]
The structure of the star enables us to explicitly calculate $S(i)$ for a family of fractional strategies $\{\vect{a_q}\}_{q \in (0,1)}$. Furthermore, this family of strategies has the property that $S(i)$ is the same continuous function of $q$ for all $i$, with $\lim_{q \to 0}S(i)=n$ and $\lim_{q \to 1}S(i)=1$. Thus, there must exist a $q \in (0,1)$ such that $S(i)=Cn/L$ for all $i$ (i.e., a fractional Nash equilibrium).

\end{proof}

In fact, we can show that if the graph is very symmetric, then there always exists a fractional Nash equilibrium.

\begin{theorem}
\label{th:vertexTransitiveNash}
Suppose $G$ is vertex-transitive. Then, for all $C \in (L/n,L)$, there exists a fractional Nash equilibrium. 
\end{theorem}

\begin{proof}[Proof Idea.]
Vertex-transitivity means that every pair of nodes $i \neq j$ are indistinguishable based on local graph structure. Then, a sufficiently ``symmetric" strategy should exhibit symmetry in the expected component sizes. Indeed, consider the strategy $\vect{a_p}$ in which every node inoculates with the same probability $p$. We prove that, under this strategy, $S(i)$ is the \emph{same} continuous function of $p \in (0,1)$ for all $i$. Thus, there must exist a $p$ such that $\vect{a_p}$ is a fractional Nash equilibrium.

\end{proof}

\subsubsection{Bounding the PoA for larger thresholds of infection}

We study the price of anarchy when the threshold of infection is higher; the adversary initially infects two different nodes, and an insecure node becomes infected if \emph{multiple} of its neighbors are infected. In particular, we prove that the price of anarchy can still be $\Omega(n)$ in this scheme.

\medskip
We first show that the price of anarchy can dramatically decrease when all thresholds are $2$. Recall that~\cite{aspnes2006inoculation} proved that the star graph $K_{1,n-1}$ has price of anarchy $\Omega(n)$ (for threshold 1). In contrast, we have the following:

\begin{theorem}
\label{th:starthreshold}
Suppose that the threshold of every node is $2$. Then, $\poa(K_{1,n-1})=O(1)$.
\end{theorem}

\begin{proof}[Proof Idea.]

Because a leaf has degree one, it can only become infected if chosen as the starting point. This means that the only node whose cost is influenced by the rest of the graph is the center. As the majority of players are effectively independent, the price of anarchy must be relatively low.

\end{proof}

However, even when all thresholds equal $2$, there are still cases where the price of anarchy is $\Omega(n)$. 

\begin{theorem}
\label{th:poathresholdLB}
If $C \geq L$, then there exists a graph $G$ for which $\poa(G) = \Omega(n)$.
\end{theorem}

\begin{proof}[Proof Idea.]

When thresholds are $1$, the star has a high price of anarchy because any node will infect the entire graph, but only one inoculation is required to split the graph into many components. The natural idea for threshold $2$ is to construct a graph in which any \emph{two} nodes will infect the entire graph, but only two inoculations are required to split the graph into many components (see Figure \ref{fig:bistar}).

\end{proof}

\subsection{Related work}

The seminal paper of~\cite{aspnes2006inoculation} introduced the inoculation game and showed constructively that every instance of the inoculation game has a pure Nash equilibrium, and some instances have many. In the same paper, it is shown that the price of anarchy for an arbitrary graph is at most $O(n)$, and that there exists a graph with price of anarchy $n/2$. Subsequent work has studied PoA on graph families such as grid graphs~\cite{moscibroda2006selfish} and expanders~\cite{kumar2010existence}.

Several works have extended the basic model of~\cite{aspnes2006inoculation} to analyze the effect of additional \emph{behaviors} on the PoA. For instance, ~\cite{moscibroda2006selfish} extend the model to include \emph{malicious} players whose goal is to maximize the cost to society. They prove that the social cost deteriorates as the number of malicious players increases, and the effect is magnified when the selfish players are \emph{unaware} of the malicious players. Somewhat conversely,~\cite{chen2010better} extend the model to include \emph{altruistic} players who consider a combination of their individual cost and the social cost (weighted by a parameter $\beta$). They prove that the social cost does indeed decrease as $\beta$ increases. Finally, ~\cite{meier2008windfall} consider a notion of \emph{friendship} in which players care about the welfare of their immediate neighbors. Interestingly, although a positive friendship factor $F>0$ is always preferable, the social cost does not necessarily decrease as $F$ increases. 

The question of how to reduce the PoA has been studied before (e.g.,~\cite{roughgarden2006severity}). For a survey regarding methods to reduce the PoA, see~\cite{roughgarden2007introduction}. The general question here is the following: how can we modify some aspect of a game to lower the PoA? To this end, variations on the \emph{infection process} (rather than the players) have also been studied; ~\cite{kumar2010existence} examine the price of anarchy in terms of the \emph{distance}, $d$, that the infection can spread from the starting point. They prove that the when $d=1$, the price of anarchy is at most $\Delta+1$, where $\Delta$ is the maximum degree of the graph. In this work, we comment on a \emph{complex contagion} extension of the model, where nodes only become infected if \emph{multiple} neighbors are infected. We show that this modification does not unilaterally decrease the PoA. There is a vast literature on complex contagion for variety of graph families~\cite{chalupa1979bootstrap,schoenebeck2016complex,eckles2018long,ebrahimi2015complex,janson2012bootstrap}. We are not aware of previous work that studies the PoA in our setting where every node has threshold 2 for infection. 


There are many classical models of epidemic spread. One of the most popular of these is the \emph{SIS model}~\cite{hethcote1973asymptotic}, which simulates infections like the flu, where no immunity is acquired after having been infected (as opposed to the \emph{SIR model}~\cite{kermack1927contribution}, in which individuals recover with permanent immunity). In this model, it was shown by~\cite{dezsHo2002halting} that the strategy which inoculates the highest-degree nodes in power-law random graphs has a much higher chance of eradicating viruses when compared to traditional strategies. It is also known that that the price of anarchy here increases as the expected proportion of high-degree nodes decreases~\cite{saha2014equilibria}. Furthermore, it is established by~\cite{wang2003epidemic, ganesh2005effect, prakash2012threshold} that epidemics die out quickly if the \emph{spectral radius} (which is known to be related to the maximum degree~\cite{cioabua2007spectral}) is below a certain threshold. This initiated the development of graph algorithms dedicated to minimizing the spectral radius by inoculating nodes~\cite{saha2015approximation}. 

\cite{chen2010better} consider the PoA in the inoculation game in graphs with maximum degree $\Delta$. They state in their paper: ``Indeed, we can show that even in the basic model of Aspnes et al. without altruism, the Price of Anarchy is bounded by $\sqrt{n \Delta}$ if all degrees are bounded by $\Delta$ (whereas the general bound is $\Theta(n))$." A similar statement is made in the PhD thesis~\cite{chen2011effects} of one of the authors of~\cite{chen2010better}. Both~\cite{chen2010better} and~\cite{chen2011effects} do not include proofs of these statements. We are not aware of a published proof of either a lower bound or an upper bound for the PoA in graph in terms of the maximum degree $\Delta$ and the number of nodes.

\section{PoA in terms of maximum degree}
\label{section:degree}

In this section we prove Theorems \ref{th:poaDeltaUB} and \ref{th:largerDeltaLB}.

\begin{remark}
As a mixed strategy is simply a distribution over pure strategies, the optimal cost can always be realized by a pure strategy. This enables the use of Theorem \ref{th:costPure} to bound the optimal cost.
\end{remark}

\begin{proof}[\textbf{Proof of Theorem \ref{th:poaDeltaUB}}]
Suppose $\gamma>0$ nodes are inoculated by strategy $\vect{a}$ and let $k_1 \leq \cdots \leq k_\ell$ denote the sizes of the connected components in $G_{\vect{a}}$. (If $\gamma=0$, then $\cost(\vect{a})=Ln$ and thus $\poa(G) = 1$.) By convexity, $\sum_{i=1}^\ell k_i^2$ is minimized when all components have the same size, $\displaystyle k_i=\frac{n-\gamma}{\ell}$ for all $i$. Thus, the optimal solution has cost at least

\begin{equation}
\label{eq:poaConvexUB}
\cost(\vect{a^*}) \geq \min\{C,L\}  \min_\gamma \left( \gamma + \frac{1}{n}\frac{\left( n-\gamma \right)^2}{\ell}\right)
\end{equation}

Note that every inoculation adds at most $\Delta$ components to the attack graph (i.e., $\ell \leq \gamma \Delta$), and (\ref{eq:poaConvexUB}) becomes

\[
\cost(\vect{a^*}) \geq  \min\{C,L\} \min_\gamma \left(\gamma \left(1+\frac{1}{\Delta n}\right) + \frac{n}{\gamma \Delta} - \frac{2}{\Delta}\right) 
\]

The function $\displaystyle f(\gamma)=\gamma \left(1+\frac{1}{\Delta n}\right) + \frac{n}{\gamma \Delta} - \frac{2}{\Delta}$ is clearly convex and attains its minimum when $\displaystyle \gamma = \frac{n}{\sqrt{1+\Delta n}}$. Substituting this value yields
\[
\cost(\vect{a^*}) \geq \min\{C,L\} \frac{2 \sqrt{n \Delta + 1}-1}{\Delta}.
\]
Corollary \ref{cor:simpleNashUB} completes the proof. Note that the $\min\{C,L\}$ term is cancelled out; the price of anarchy bound is independent of $C,L$.

\end{proof}

\begin{proof}[\textbf{Proof of Theorem \ref{th:largerDeltaLB}}]
Consider an arbitrary $\Delta$-regular graph $G$ on $m=2\sqrt{n/\Delta}$ vertices $v_1,\cdots,v_m$. We construct a new graph $G'$ with $n$ vertices by replacing the edges of $G$ with (disjoint) paths of length $\displaystyle \ell=(n-m)/ (m\Delta/2) = \sqrt{n/\Delta}-2/\Delta$. (Round path lengths such that the total number of nodes is $n$.)

Consider the strategy which secures $v_1,\cdots,v_m$. The inoculations cost $C m$ and create $m \Delta / 2$ components in $G'$ of size $\ell$. This strategy upper bounds the cost of the optimal strategy by
\[
C m + L (n-m) \frac{\ell}{n} = 2C\sqrt{n/\Delta} + 2L \frac{\left(n-2\sqrt{n/\Delta}\right)^2}{\sqrt{n^3\Delta}} = O(\sqrt{n/\Delta}).
\]
As $C \geq L$, the worst case Nash equilibrium has cost $Ln$. 

\end{proof}

\begin{figure}[ht]
    \centering
    \includegraphics[width = 0.3 \textwidth]{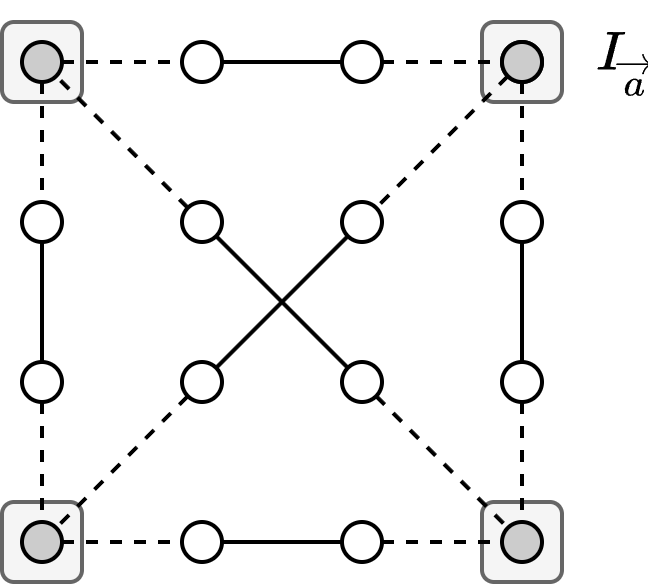}
    \caption{The inoculation strategy for this instance the graph constructed in Theorem \ref{th:largerDeltaLB} inoculates $4$ nodes, creating six disjoint paths of size $2$. This has cost significantly lower than the worst Nash equilibrium which does not inoculate at all.}
    \label{fig:pathclique}
\end{figure}

\section{Existence of fractional equilibria}

In this section we prove Theorems \ref{th:starNash} and \ref{th:vertexTransitiveNash}. 

\begin{proof}[\textbf{Proof of Theorem \ref{th:starNash}}]
Let $\vect{a}$ be the strategy in which every leaf inoculates with probability $p$ and the root inoculates with probability $q$. Then,
\begin{itemize}
    \item $S(\mathrm{root}) = 1 + (1-p)(n-1)$, 
    \item $S(\mathrm{leaf}) = q + (1-q)[2 + (1-p)(n-2)]$.
\end{itemize}
It is easy to verify that $\displaystyle S(\mathrm{root}) = S(\mathrm{leaf}) = \frac{n-q}{(n-2)q+1}$ when $\displaystyle p=\frac{(n-1)q}{(n-2)q+1}$. The former expression is a continuous function of $f(q) : [0,1] \to [1,n]$, with $f(0)=n$ and $f(1)=1$. Thus, for all $\displaystyle  C \in (L/n, L)$, there exists a $q \in (0,1)$ such that $\displaystyle S(\mathrm{root}) = S(\mathrm{leaf}) =\displaystyle \frac{Cn}{L}$ (i.e., $\vect{a}$ is a fractional Nash equilibrium). 

\end{proof}

\begin{proof}[\textbf{Proof of Theorem \ref{th:vertexTransitiveNash}}]
We prove that there exists a $p \in (0,1)$ such that the strategy $\vect{a_p} = p^n$ is a fractional Nash equilibrium. By the definition of vertex-transitivity, for any two nodes $i \neq j$, there exists an automorphism $f : G \to G$ such that $f(i)=j$. 

\medskip
Consider an arbitrary set of inoculated nodes, $A$, and their image, $f(A)$. By definition, 
\[
\Pr(I_{\vect{a}}=A) = \Pr(I_{\vect{a}}=f(A)) = p^{|A|}(1-p)^{n-|A|}
\]
Let $S(i|A)$ denote the size of the connected component containing $i$ in the attack graph $G \setminus A$. By vertex transitivity, $S(i|A) = S(f(i)|f(A)) = S(j|f(A))$. Then, by linearity of expectation,
\begin{align*}
S(i) &= \sum_{A} S(i|A) \Pr(I_{\vect{a}}=A) \\
&= \sum_{A} S(j|f(A)) \Pr(I_{\vect{a}}=f(A)) = S(j).
\end{align*}
Also note that $S(i)$ is a polynomial in $p$ satisfying $\lim_{p \to 0} S(i) = n$ and $\lim_{p \to 1} S(i) = 1$. Therefore, for all $\displaystyle  C \in (L/n, L)$, there exists a $p \in (0,1)$ such that $S(i) = Cn/L$ for all $i$.

\end{proof}

\section{Larger thresholds of infection}

In this section we prove Theorems \ref{th:starthreshold} and \ref{th:poathresholdLB}.

\begin{proof}[\textbf{Proof of Theorem \ref{th:starthreshold}}]

Any leaf node has only one neighbor and can only be infected if chosen at the start. Thus,
\[
p_{\mathrm{leaf}}(\vect{a})=\frac{n-1}{\binom{n}{2}} = \frac 2n,
\] 
and $\displaystyle \cost_{\mathrm{leaf}}(\vect{a}) = C a_i + \frac{2L (1-a_i)}{n}$. Then, for large enough $n$, no leaf node will inoculate in a Nash equilibrium. Therefore, the worst case Nash equilibrium has cost at most
\[
\max\{L,C\} + L \cdot (n-1) \frac{2}{n}= \max\{L,C\} + 2L\left(1 - \frac{1}{n}\right)
\]
Now consider the optimal strategy. If at least one node inoculates with probability $a_i \geq 1/2$, then $\cost(\vect{a^*}) \geq C/2$. Otherwise, the root (if insecure) is infected with probability at least $1/4$; either it is chosen at the start, or two insecure nodes are chosen. As the root inoculates with probability at most $1/2$, the optimal social cost is at least $L/8$.

\end{proof}

\begin{proof}[\textbf{Proof of Theorem \ref{th:poathresholdLB}}]

Consider the graph $G = K_{2,n-2}$ with an edge between the two nodes on the smaller side. A Nash equilibrium where every node chooses not to inoculate has cost $Ln$ as every two nodes will infect the entire graph. On the other hand, the strategy which inoculate both nodes on the smaller side upper bounds the cost of the social optimum by
\[
\cost(\vect{a^*}) \leq 2C + (n-2) \cdot L \cdot \frac{2}{n} \leq 2(C+L)
\]
Hence $\poa(G) = \Omega(n)$, concluding the proof. 

\end{proof}

\begin{figure}[ht]
    \centering
    \includegraphics[width = 0.25 \textwidth]{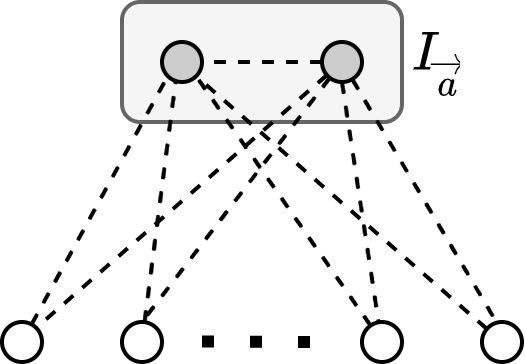}
    \caption{Inoculating the two central nodes of the graph constructed in Theorem \ref{th:poathresholdLB} is ``efficient" as it splits the rest of the graph into $n-2$ singleton components. This has a much lower social cost than the Nash equilibrium in which no node inoculates.} 
    \label{fig:bistar}
\end{figure}

\section*{Acknowledgements}
We are grateful to James Aspnes, David Kempe and Ariel Procaccia for useful comments. We thank the anonymous reviews for their feedback and suggesting a simpler proof of Theorem~\ref{th:largerDeltaLB}. Part of this work was done while the third author was visiting the Simons Institute for the Theory of Computing. Their hospitality is greatly acknowledged. 

\bibliographystyle{alpha}
\bibliography{biblio}

\vfill
\pagebreak

\appendix

\section{Other families}
\label{section:other}

We also obtain price of anarchy bounds for planar graphs, trees, expanders and Erdős-Renyi random graphs. 

\begin{theorem}
\label{th:planarLB}
Let $G$ be a planar graph. If $C \geq L$, then $\poa(G)=\Omega(n^{1/3})$.
\end{theorem}

This is the best lower bound possible for planar graphs, as it was shown in~\cite{moscibroda2006selfish} that a (planar) two-dimensional grid has price of anarchy $O(n^{1/3})$. Clearly, planar graphs can have price of anarchy $\Omega(n)$, as the classic star graph example is planar. Even for bounded degree planar graphs, Theorem \ref{th:largerDeltaLB} showed that the lower bound can be as large as $\Omega(\sqrt{\Delta n})$. Next, we apply a similar method to trees, obtaining a stronger lower bound for this type of planar graph.

\begin{theorem}
    \label{th:treeLB}
    Let $T$ be a tree. If $C \geq L$, then $\poa(T) = \Omega(\sqrt{n})$.
\end{theorem}

Finally, we show that the price of anarchy for well-connected graphs is $O(1)$. These are graphs with the property that even the optimal strategy has cost $\Omega(n)$. That is, the situation requires that a majority of individuals either inoculate or get infected. 

\begin{theorem}
\label{th:randomPOA}
Let $G=G(n,p)$ for $\displaystyle p \geq \frac{1+\epsilon}{n}$ with $\epsilon>0$. Then, $\poa(G) = O(1)$ with high probability. 
\end{theorem}

\subsection{Planar graphs}

Before proving Theorem \ref{th:planarLB}, we state the \emph{planar separator theorem}:

\begin{theorem}[\cite{lipton1979separator}]
\label{th:planarsep}
Let $G$ be a planar graph on $n$ nodes. There exists a partition of the nodes $V = A \cup S \cup B$ such that 
\begin{itemize}
    \item $|A|, |B| \leq n/2$
    \item $|S| = O(\sqrt{n})$
    \item There does not exist an edge between $A$ and $B$.
\end{itemize}
\end{theorem}

\begin{proof}[\textbf{Proof of Theorem \ref{th:planarLB}}]
By Theorem \ref{th:planarsep}, any planar graph $G$ can be split into $2$ components of size at most $n/2$ by removing at most $O(\sqrt{n})$ nodes. Iteratively, these components can both be split into $4$ total components of size at most $n/4$ by removing $2 \cdot O(\sqrt{n/2})$ more nodes. Generally, the number of inoculations required to create $\ell:=2^k$ components of size at most $2^{-k}n$ is upper bounded by

\[
\sum_{i=1}^k 2^i \cdot O(\sqrt{n 2^{-i}}) = O(\sqrt{n}) \cdot \sum_{i=1}^k 2^{i/2} = O(\sqrt{n \ell}).
\]
By optimizing over $\ell$, we have that for any planar graph,

\[
\cost(\vect{a^*}) \leq C \cdot O(\sqrt{n \ell}) + \frac{L}{n} \cdot \ell \cdot \left(\frac{n}{\ell}\right)^2 = O\left(C\sqrt{n\ell} + L \frac{n}{\ell} \right) = O(n^{2/3})
\]

\end{proof}

\subsection{Trees}

We obtain an analogous lower bound for trees by a slightly more efficient separator theorem.

\begin{theorem}[\cite{jordan1869assemblages}]
\label{th:jordanTree}
In any tree on $n$ nodes, there exists a node whose removal leaves components of size at most $n/2$.
\end{theorem}

\begin{lemma}
\label{lemma:trees}
Let $T$ be a tree. $T$ can be split into components of size at most $\sqrt{n}$ by removing $O(\sqrt{n})$ nodes.
\end{lemma}

\begin{proof}
By Theorem \ref{th:jordanTree}, $T$ can be split into multiple components of size at most $n/2$ by removing one node. Each of these components, in turn, can be split into components of size at most $n/4$ by removing one node each. This process can be repeated until all components have at most $\sqrt{n}$ nodes. Figure~\ref{fig:randomTree} illustrates this algorithm for a particular tree. We now prove that this process requires $O(\sqrt{n})$ removals.

Let $N_i$ denote the number of components in the attack graph with size greater than $n\cdot 2^{-i}$ but at most $n \cdot 2^{-i+1}$. Removing a node from such a component (per Theorem \ref{th:jordanTree}) will decrease $N_i$ by one and increase some $N_{j}$'s where $j>i$. Note that $N_i \leq 2^i$ as component sizes sum to at most $n$. Therefore, the number of removals required to reduce $N_i$ to $0$ for all $i \leq \log_2(\sqrt{n})$ is upper bounded by
\[
\sum_{i=1}^{\log_2 (\sqrt{n})} 2^{i} = 2\sqrt{n} - 2.
\]

\end{proof}

\begin{figure}[ht]
    \centering
    \includegraphics[width=0.35 \textwidth]{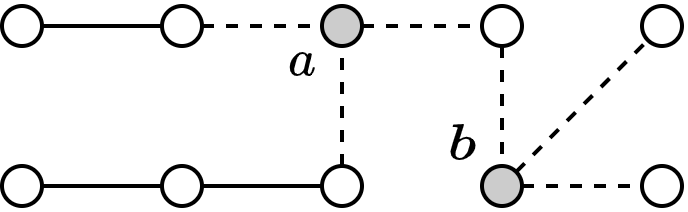}
    \caption{Following the algorithm given by Lemma~\ref{lemma:trees} for this tree, we inoculate node $a$, and then node $b$. If $C \geq L$, then this strategy achieves the lower bound on PoA given by Theorem~\ref{th:treeLB}.} 
    \label{fig:randomTree}
\end{figure}

\begin{proof}[\textbf{Proof of Theorem \ref{th:treeLB}}]
By Lemma~\ref{lemma:trees}, there exists a a strategy that inoculates $O(\sqrt{n})$ nodes, leaving attack graph components of size at most $\sqrt{n}$. This leaves $n-O(\sqrt{n}) \leq n$ insecure nodes whose probability of infection is at most $\sqrt{n}/n=1/\sqrt{n}$. Thus, the cost of this strategy upper bounds the social optimum by
\[
\cost(\vect{a^*}) \leq C \cdot O(\sqrt{n}) + L n \cdot \frac{1}{\sqrt{n}} \leq O(\sqrt{n}).
\]

\end{proof}

\subsection{Random graphs}

The proof of Theorem \ref{th:randomPOA} amounts to the following two lemmas regarding \emph{expander} graphs.

\begin{definition}
An $n$-node graph $G=(V,E)$ is called an $\alpha$-expander if for every subset $S$ of nodes with $|S| \leq n/2$, there are at least $\alpha |S|$ nodes in $V \setminus S$ neighboring a node in $S$. (Here, $\alpha>0$ is a fixed constant independent of $n$.)
\end{definition}

The following lemma may be folklore, but we include the proof here for completeness.

\begin{lemma}
\label{lemma:dismanteling}
Suppose $G$ is an $\alpha$-expander. Then there exists constants $\delta,\epsilon>0$ such that for every subset $A$ of $\delta n$ nodes, removing $A$ from $G$ leaves a connected component of size at least $\epsilon n$. 
\end{lemma}

\begin{proof}
Suppose that there is a subset of nodes $A$ of size $\alpha n/4$ such that removing $A$ results with connected components $S_1 \ldots S_m$ all of size smaller than $\epsilon n$ where $\epsilon>0$ is a constant to be determined later. This implies that there is a subset $S$ of nodes satisfying 
\[
\left( \frac12 - \epsilon \right)n\leq |S| \leq \frac{n}{2}
\]
with at most $\alpha n/4$ neighbors in $V \setminus S$. 

Indeed, initialize $S$ as the empty set and keep adding components $S_i$ (that were not added already) until $S$ has size at least $(1/2-\epsilon)n$. Note that $S$ has cardinality no larger than $n/2$ and has at most $\alpha n /4$ neighbors. For $\epsilon$ sufficiently small, this is a contradiction to $G$ being an $\alpha$-expander.

\end{proof}

\begin{remark}
Observe that the assertion in Lemma \ref{lemma:dismanteling} applies also if $G$ contains a sub-graph with $\Omega(n)$ nodes which is an $\alpha$-expander.
\end{remark}

\begin{lemma}
\label{lemma:PoA expanders}
For a graph $G$, suppose there exists constants $\delta,\epsilon>0$ such that for every subset $A$ of $\delta n$ nodes, removing $A$ from $G$ leaves a connected component of size at least $\epsilon n$. Then, $\poa(G)=O(1)$. 
\end{lemma}

\begin{proof}
If removing $\delta n$ nodes leaves a component of size at least $\epsilon n$, then 
\begin{itemize}
    \item If $|I_{\vect{a^*}}| \leq \delta n$, then $\cost(\vect{a^*}) \geq L\epsilon^2n$.
    \item If $|I_{\vect{a^*}}| > \delta n$, then $\cost(\vect{a^*}) > C \delta n$.
\end{itemize}
By Corollary \ref{cor:simpleNashUB}, we have

\[
\poa(G) \leq \frac{\min\{C,L\}n}{\min\{C \delta, L \epsilon^2\}n} = O(1).
\]

\end{proof}

\begin{proof}[\textbf{Proof of Theorem \ref{th:randomPOA}}]
Note that $G(n,p)$ may not be an $\alpha$-expander for fixed $\alpha>0$: If edge probability $p=(1+\epsilon)/n$, then $G(n,p)$ contains an induced path of length $\Omega(\log n)$ with high probability~\cite{benjamini2014mixing}. Nevertheless, it is known that if $p \geq (1+\epsilon)/n$, then $G(n,p)$ contains an $\alpha$-expander with $\Omega(n)$ nodes with high probability (see~\cite{krivelevich2019expanders}). Lemmas \ref{lemma:dismanteling} and \ref{lemma:PoA expanders} complete the proof.

\end{proof}

\end{document}